\documentclass[11pt]{article}
\usepackage{amsmath,amssymb,amsfonts,amsthm,epsfig}
\usepackage[usenames,dvipsnames]{xcolor}
\usepackage{bm,xspace}
\usepackage{tcolorbox}
\usepackage{cancel}
\usepackage{fullpage}
\usepackage{liyang}
\usepackage{framed}
\usepackage{verbatim}
\usepackage{enumitem}
\usepackage{array}
\usepackage{multirow}
\usepackage{afterpage}
\usepackage{mathrsfs}
\usepackage{pifont} 
\usepackage{chngpage}
\usepackage[normalem]{ulem}
\usepackage{boxedminipage}
\usepackage{caption}
\usepackage[T1]{fontenc}

\def\colorful{1}

\ifnum\colorful=1

\fi
\ifnum\colorful=0

\fi

\usepackage{algorithm}
\floatname{algorithm}{Procedure}
\usepackage[noend]{algpseudocode}
\usepackage{float}

\newcommand{\pparagraph}[1]{\bigskip \noindent {\bf {#1}}}

\newcommand{\yes}{\mathtt{YES}}
\newcommand{\no}{\mathtt{NO}}
\newcommand{\A}{\mathcal{A}}
\newcommand{\D}{\mathcal{D}}
\newcommand{\mpc}{\mbox{\textsf{MPC}}}
\newcommand{\ampc}{\mbox{\textsf{AMPC}}}
\newcommand{\ceil}[1]{\lceil #1 \rceil}

\newcommand{\OCTC}{\textsc{1v2-Cycle}} 
\newcommand{\OCKC}{\textsc{1vk-Cycle}} 
\newtheorem*{theorem*}{Theorem}

\begin{document}

\title{New lower bounds for Massively Parallel Computation \\
from query complexity}

\author{Moses Charikar \and \hspace{-12pt} Weiyun Ma \vspace{10pt} \\
\hspace{-17pt} {\small {\sl Stanford University}}
\and Li-Yang Tan}

\date{\small{\today}}

\maketitle

\begin{abstract}
Roughgarden, Vassilvitskii, and Wang {\footnotesize{(JACM 18)}} recently introduced a novel framework for proving lower bounds for Massively Parallel Computation using techniques from boolean function complexity.  We extend their framework in two different ways, to capture two common features of Massively Parallel Computation: 
\begin{itemize}
\item[$\circ$] \textsl{Adaptivity}, where machines can write to and adaptively read from shared memory throughout the execution of the computation.   Recent work of Behnezhad et al.~{\footnotesize{(SPAA 19)}} showed that adaptivity enables significantly improved round complexities for a number of central graph problems.

\item[$\circ$] \textsl{Promise problems}, where the algorithm only has to succeed on certain inputs.  These inputs may have special structure that is of particular interest, or they may be representative of hard instances of the overall problem. \end{itemize} 

Using this extended framework, we give the first unconditional lower bounds on the complexity of distinguishing whether an input graph is a cycle of length $n$ or two cycles of length $n/2$.  This promise problem, $\OCTC$, has emerged as a central problem in the study of Massively Parallel Computation.  We prove that any adaptive algorithm for the $\OCTC$ problem with I/O capacity $O(n^{\eps})$ per machine requires $\Omega(1/\eps)$ rounds, matching a recent upper bound of Behnezhad et al.

In addition to strengthening the connections between Massively Parallel Computation and boolean function complexity, we also develop new machinery to reason about the latter.   At the heart of our proofs are optimal lower bounds on the query complexity and approximate certificate complexity of the $\OCTC$ problem.  
\end{abstract}

\thispagestyle{empty}

\newpage
\setcounter{page}{1}

\section{Introduction} 
In recent years, there has been a surge of effective parallel computation platforms for processing large-scale data. Examples include MapReduce~\cite{Dean, Dean2}, Spark~\cite{Zaharia}, and Hadoop~\cite{White}. As a theoretical counterpart, the {\sl Massively Parallel Computation} ($\mpc$) model~\cite{Karloff, Goodrich, Andoni, Beame} has been proposed to capture the central features shared by these practical platforms, and is by now a standard framework in the modern study of parallelism.  A partial listing of works on the $\mpc$ model appearing within the past two years includes:~\cite{Ahn, Andoni2, Boroujeni, Brandt, Czumaj, Ghaffari, Harvey, Onak, Roughgarden, Yaroslavtsev, Assadi2, Assadi4, Assadi3, Andoni3, Behnezhad2, Behnezhad, Behnezhad3, Gamlath, Ghaffari4, Ghaffari3, Ghaffari2, Lacki}. Most relevant to this work, a significant amount of attention has been devoted to graph problems, especially graph connectivity and its variants~\cite{Karloff, Rastogi, Beame, Andoni2, Roughgarden, Assadi3, Andoni3, Behnezhad2, Behnezhad}.

In an $\mpc$ computation, an input of length $N$ is partitioned (arbitrarily) and distributed to a collection of machines. The computation proceeds in synchronous rounds. In each round, each machine performs a computation on the messages it receives from the previous round, and then communicates the results of its computation to other machines as input for the subsequent round.  An important feature of the $\mpc$ model is that no restrictions are placed on the computational power of each machine: its messages to the other machines in the next round are an arbitrary function of the messages it receives from the previous round.  The only restriction is on the {\sl I/O capacity} of each machine: in each round, the total size of the messages that any machine receives or sends is at most $S$ bits, where $S$ is smaller than $N$.\footnote{I/O capacity is commonly also referred to as ``space'' in the $\mpc$ literature.} The principal complexity measure in the $\mpc$ model is the {\sl number of rounds} it takes to finish the computation.

\pparagraph{Our contributions.}  Roughgarden, Vassilvitskii, and Wang~\cite{Roughgarden} recently introduced a novel framework for proving $\mpc$ lower bounds using techniques from boolean function complexity.  We extend their framework in two different ways, to capture two common features of Massively Parallel Computation: {\sl adaptivity} and {\sl promise problems}.   Using our extended framework, we give an unconditional, optimal lower bound on the complexity of distinguishing whether an input graph is a cycle of length $n$ or two cycles of length $n/2$ in the adaptive $\mpc$ model of Behnezhad, Duhlipala, Esfandiari, \L\k{a}cki, Shudy, and Mirrokni~\cite{Behnezhad}.  This is a promise problem that has emerged as a central problem in the study of Massively Parallel Computation as it captures an essential bottleneck in the design of efficient graph algorithms.  In addition to strengthening the connections between Massively Parallel Computation and boolean function complexity, we also develop new machinery to reason about the complexity of boolean functions.   

\subsection{Background and motivation}
\label{sec:background}
\vspace{-8pt} 

\pparagraph{The power of $\mpc$ computation.}  While there have been a great number of works designing efficient $\mpc$ algorithms for various problems, there have been significantly fewer {\sl hardness results}---namely, lower bounds on round complexity in the $\mpc$ model.   Intuitively, there should be close parallels between round complexity in the $\mpc$ model and  {\sl depth complexity} in the study of boolean circuits.  The latter has long been a major focus of research in circuit complexity, and by now a range of techniques has been developed for proving depth lower bounds for various types of circuits.  However, compared to circuits, the power and generality of the $\mpc$ model makes proving lower bounds significantly more challenging. The following two central features of the $\mpc$ model, already alluded to above, exemplify this contrast:

\label{two-features} 
\begin{itemize}[leftmargin=0.5cm]
\item[$\circ$] {\sl Input-dependent communication pattern}: the communication pattern among the machines---which machines send messages to which in a given round---can depend on the outcome of each machine's computation (and hence can differ for different inputs).  In circuit complexity, on the other hand, the topology of the circuit---which gates are connected to which via wires---is fixed and the same for all possible inputs. 

\item[$\circ$] {\sl Computationally unbounded machines:}  In each round, each machine can compute an arbitrary function on the message it receives from the previous round. Equivalently, we may view each machine as an arbitrary function $M : \zo^S \to \zo^S$, with no constraints on its computational complexity.  In circuit complexity, on the other hand, the primary focus is on circuits composed of computationally simple gates (such as {\sc And}, {\sc Or}, {\sc Not}, and {\sc Majority}). There has been some work on circuits comprising gates that compute arbitrary functions, but lower bounds against such circuits have been notoriously difficult to prove (see e.g.~\cite{Valiant} and Chapter 13 of \cite{Jukna}).
\end{itemize} 

\vspace{-10pt} 

\pparagraph{The~\cite{Roughgarden} framework: lower bounds via the polynomial method.}  Recent work of Roughgarden, Vassilvitskii, and Wang~\cite{Roughgarden} opens up a new avenue towards proving $\mpc$ lower bounds.  Their work introduces a simple and elegant model for $\mpc$ computation that captures the key features discussed above (input-dependent communication and computationally unbounded machines), and they draw a connection between lower bounds in this model and boolean function complexity \cite{Jukna}.  Specifically, they show that functions computable by efficient $\mpc$ algorithms can be represented as {\sl low-degree polynomials}.  This allows them to leverage a large body of work and techniques on the complexity of polynomial representations, often referred to as ``the polynomial method'' in complexity theory (see e.g.~\cite{Bei93,Aar08,Wil14}), to prove lower bounds on round complexity in the $\mpc$ model.

\label{RVW-framework}
In more detail,~\cite{Roughgarden} shows that a function $g : \zo^N \to \zo$ that is computed by an $R$-round $\mpc$ algorithm using machines with I/O capacity $S$ can be represented by a polynomial of degree $S^R$ (over the reals).  Their techniques extend to randomized algorithms, in which case $g$ has an {\sl approximate} polynomial representation of degree $S^R$: a polynomial $p$ such that $|p(x) - g(x)| \le \frac1{3}$ for all $x \in \zo^N$.  As the main application of their framework and techniques,~\cite{Roughgarden} give the first lower bounds on the $\mpc$ round complexity of basic graph connectivity problems. In particular, they prove the following lower bound for deciding the connectivity of undirected $n$-node graphs:\footnote{\cite{Roughgarden} applies their framework to four graph connectivity problems: undirected connectivity, undirected st-connectivity, and their directed versions.  For concreteness, we will focus on the simplest case of undirected connectivity.}

\hypertarget{RVW}{}
\begin{theorem*}[\cite{Roughgarden}'s lower bound for {\sc Connectivity}]  
Any $\mpc$ algorithm for {\sc Connectivity} using machines with I/O capacity $S = n^{\eps}$ requires $\Omega(1/\eps)$ rounds. 
\end{theorem*}


Qualitatively, this shows that the round complexity of {\sc Connectivity} has to scale with the I/O capacity of the machines.\footnote{We note that there have been a number of works in the $\mpc$ literature focusing on algorithms with round complexities that are an absolute constant, independent of I/O capacity (e.g.~\cite{Afrati, Koutris}).}   This lower bound also implies, for example, that if the I/O capacity is subpolynomial in $n$ (i.e.~$S = n^{o(1)}$), then solving {\sc Connectivity} requires a superconstant number of rounds.   A notable strength of~\cite{Roughgarden}'s lower bound is that it is independent of the number of machines: such a lower bound holds even if the algorithm is allowed an exponential number of machines in each round.

At the heart of~\cite{Roughgarden}'s proof are results on the complexity of representing {\sc Connectivity} as polynomials: in the deterministic case, they prove an ${n\choose 2}$ lower bound on the degree of {\sc Connectivity}, and in the randomized case, they prove an $\Omega(n^{1/3})$ lower bound on the approximate degree of {\sc Connectivity}:

\begin{fact}
\label{fact:connectivity-degree} 
$\deg({\textsc{Connecitivity}}) \ge {n\choose 2}$ and $\wt{\deg}(\textsc{Connectivity}) \ge \Omega(n^{1/3})$.
\end{fact}

\pparagraph{The $\OCTC$ problem and the logarithmic round conjecture.}  
\hyperlink{RVW}{\cite{Roughgarden}'s lower bound} for {\sc Connectivity} is not known to be tight.  It is widely believed that for machines with I/O capacity $S = n^{\eps}$, the number of rounds required is actually $\Omega_\eps(
\log n)$, i.e.~logarithmically many rounds for constant~$\eps$~\cite{Karloff, Rastogi, Beame, Andoni2, Yaroslavtsev, Assadi3, Behnezhad2}.   

In fact, such a lower bound is conjectured to hold even for the simpler promise problem of distinguishing whether an input graph is a cycle of length $n$ or two cycles of length $n/2$ \cite{Yaroslavtsev}.  Formally, the partial function $\OCTC: \Delta_\OCTC \to \zo$ is defined on the domain $\Delta_\OCTC \subset \zo^{{n \choose 2}}$, which consists of all $n$-node graphs that is either a cycle of length $n$ or two disjoint cycles each of length $n/2$, and 
\[ \OCTC(G) = 
\begin{cases}
1 & \text{if $G$ is a cycle of length $n$} \\
0 & \text{if $G$ is two disjoint cycles each of length $n/2$.}
\end{cases}    \]

We observe that {\sc Connectivity} extends $\OCTC$,\footnote{Meaning that the $1$-inputs of $\OCTC$ are a subset of the $1$-inputs of {\sc Connectivity}, and likewise for the $0$-inputs.} so indeed it can only be easier to solve the $\OCTC$ problem (or equivalently, lower bounds against $\OCTC$ yield lower bounds against {\sc Connectivity}).  As mentioned above, it has been conjectured that {\sc Connectivity} requires logarithmically many rounds even when restricted to the promise instances of $\OCTC$; we call this the ``Logarithmic-round $\OCTC$ conjecture:''

\hypertarget{OCTC-conjecture}{} 
\begin{conjecture}[Logarithmic round $\OCTC$ conjecture \cite{Yaroslavtsev}]
\label{OCTC-conjecture}
Any $\mpc$ algorithm for the $\OCTC$ problem using machines with I/O capacity $S = n^{\eps}$ requires $\Omega_\eps(\log n)$ rounds. 
\end{conjecture}

This is by now a widely accepted conjecture in the $\mpc$ literature. Based on this conjecture, a number of works have shown conditional hardness results for a variety of problems in $\mpc$~\cite{Yaroslavtsev, Andoni3, Behnezhad2, Ghaffari4, Lacki}.

\pparagraph{The adaptive $\mpc$ model of~\cite{Behnezhad}.} Very recently, Behnezhad, Duhlipala, Esfandiari, \L\k{a}cki, Shudy, and Mirrokni~\cite{Behnezhad} introduced an {\sl adaptive} extension of the $\mpc$ model, which they call the $\ampc$ model. Motivated by the practical success of the low-latency remote direct memory access framework, the natural notion of {\sl adaptivity} allows machines to adaptively query a shared memory that stores all messages produced in the previous round. In this setting, each message consists of a constant number of words, and the analogue of I/O capacity becomes the following: in any round, each machine can query for at most $S$ messages from the shared memory of the previous round and write at most $S$ messages to that of the current round.~\Cref{fig:TwoModelsComparison} illustrates how adaptivity changes the way in which a machine participates in the computation.

\begin{figure}[htb]
\begin{center}
\includegraphics[scale=.42]{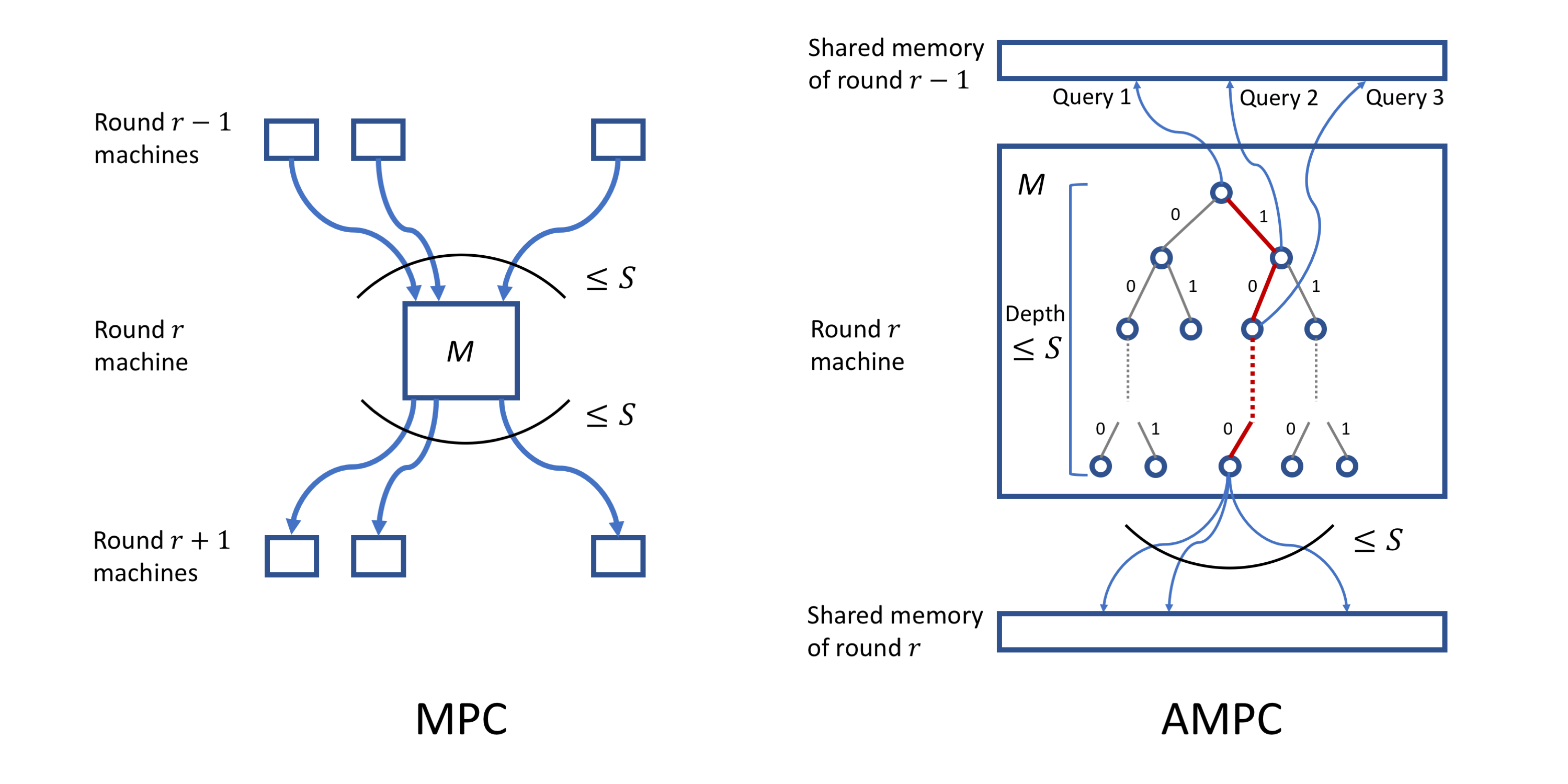}
\end{center}
\vspace{-20pt}
  \captionsetup{width=.9\linewidth}
\caption{A comparison between the computation on a machine $M$ in the (non-adaptive) $\mpc$ model (on the left) and the $\ampc$ model (on the right). In $\mpc$, $M$ ``passively'' receives at most $S$ bits and computes a function on them. In $\ampc$, the computation that $M$ performs follows a decision tree of depth at most $S$: $M$ ``actively'' queries the shared memory, and later queries may depend on the outcomes of earlier ones.}
\label{fig:TwoModelsComparison}
\end{figure}
\medskip

As one would expect, the $\ampc$ model is stronger than the $\mpc$ model: any $\mpc$ computation can be simulated by an $\ampc$ computation using the same number of rounds. Behnezhad et al.~gave $\ampc$ algorithms for a number of central graph problems with round complexities substantially lower than those of the best known $\mpc$ algorithms. In particular, they showed that the Logarithmic-round $\OCTC$ Conjecture (\Cref{OCTC-conjecture}) does {\sl not} hold in the adaptive setting: 

\hypertarget{spaa-upper-bound}{}
\begin{theorem*}[\cite{Behnezhad}'s adaptive algorithm for $\OCTC$] 
For any $\eps < 1$ there is a randomized $\ampc$ algorithm solving the $\OCTC$ problem in $O(1/\eps)$ rounds using machines with I/O capacity $S = n^{\eps}$. Consequently, the Logarithmic-round $\OCTC$ Conjecture (\Cref{OCTC-conjecture}) does not hold in the $\ampc$ model. 
\end{theorem*}

\subsection{This work} 

This paper is very much inspired by the works of Roughgarden et al.~\cite{Roughgarden} and Behnezhad et al.~\cite{Behnezhad}.  Just like them, our goal is to contribute to the theoretical understanding of modern massively parallel computing systems, with a particular emphasis on lower bounds and impossibility results.  This is a timely but challenging research direction, since the landscape of $\mpc$ systems is rapidly evolving and is largely influenced by and dependent on practice.  

Our contributions are twofold.  First, we extend the lower bound framework of~\cite{Roughgarden} in two different ways, to capture two common features of Massively Parallel Computation: 
\begin{itemize}[leftmargin=0.5cm]
\item[$\circ$] \textsl{Adaptivity}, where machines can write to and adaptively read from shared memory throughout the execution of the computation.  As discussed above, the power of adaptivity in Massively Parallel Computation was recently highlighted in~\cite{Behnezhad}, who showed that it enables significantly improved round complexities for a number of central graph problems. 

\item[$\circ$] \textsl{Promise problems}, where the algorithm only has to succeed on certain inputs.  These inputs may have special structure that is of particular interest (e.g.~such structure may be common in practical instances), or they may be representative of hard instances of the overall problem. \end{itemize} 

Like the two features discussed at the beginning of~\Cref{sec:background},  we believe that these are important features of modern Massively Parallel Computation, which call for theoretical frameworks---and lower bound techniques---that capture them. 

Complementing this, as our second contribution, we  apply our  framework to give an unconditional, optimal lower bound on the round complexity of the $\OCTC$ problem in the $\ampc$ model:

\begin{theorem}[Optimal $\ampc$ lower bound for $\OCTC$]
\label{thm:Main}
Any $\ampc$ algorithm for the $\OCTC$ problem using machines with I/O capacity $S = n^{\eps}$ requires $\Omega(1/\eps)$ rounds. 
\end{theorem}

\Cref{thm:Main} is optimal as it matches the \hyperlink{spaa-upper-bound}{upper bound of \cite{Behnezhad}}.   It extends \hyperlink{RVW}{\cite{Roughgarden}'s lower bound} in two ways: first, it applies to adaptive computations while \hyperlink{RVW}{\cite{Roughgarden}'s lower bound} only applies to non-adaptive computations, and second, the $\OCTC$ problem is a restriction of {\sc Connectivity} to specific promise instances.  Prior to our work, neither extension (even on their own) was known. 

Like the lower bound of~\cite{Roughgarden},~\Cref{thm:Main} holds regardless of the number of machines involved in the computation.

\subsubsection{Extensions and implications}
\label{sec:extensions} 
In fact, our method can be applied to proving lower bounds for the more general promise problem $\OCKC$ of distinguishing between a cycle of length $n$ versus $k$ cycles of length $\frac{n}{k}$, where $k$ divides $n$.
\begin{theorem}[Optimal $\ampc$ lower bound for $\OCKC$]
\label{thm:OCKC}
For $k = O(n^\delta)$ with $\delta \in (0,1)$, any $\ampc$ algorithm for the $\OCKC$ problem using machines with I/O capacity $S = n^{\eps}$ requires $\Omega(1/\eps)$ rounds. 
\end{theorem} 
One easily checks that the $O(1/\epsilon)$-round $\ampc$ algorithm of \cite{Behnezhad} for $\OCTC$ can be generalized to an $O(1/\epsilon)$-round $\ampc$ algorithm for $\OCKC$. Thus our bound is indeed optimal.

Moreover, our unconditional $\Omega(1/\epsilon)$-round lower bound for $\OCTC$ (\Cref{thm:Main}) can be used to convert hardness results conditioned on the Logarithmic-round $\OCTC$ Conjecture (\Cref{OCTC-conjecture}) into (weaker) unconditional hardness results in the $\ampc$ model.  For instance, Yaroslavtsev and Vadapalli \cite{Yaroslavtsev} study the \textsc{$k$-Single-Linkage Clustering} (\textsc{$k$-SLC}) problem: given $n$ vectors in $\R^d$, partition them into $k$ clusters so as to maximize the minimum distance between two vectors that belong to different clusters. Conditioned on~\Cref{OCTC-conjecture}, they show that any $o_\epsilon(\log n)$-round $\mpc$ algorithm cannot approximate \textsc{2-SLC} within a factor of $1.84-\delta$ for $d = \Omega(\log n/\delta^2)$ under $\ell_2^d$, or within a factor of $3$ for $d = \Omega(n)$ under $\ell_0^d$ or $\ell_1^d$ \cite[Theorem 3.3]{Yaroslavtsev}.  We remark that their reduction from $\OCTC$ to \textsc{2-SLC} requires only a constant number (independent of $\epsilon$) of rounds in $\mpc$ (and thus in $\ampc$) regardless of whether the input graph is represented with its adjacency matrix or adjacency list.  Therefore,~\Cref{thm:Main} implies the following unconditional hardness of approximation result:

\begin{theorem}[Unconditional hardness of \textsc{2-SLC} in $\ampc$]
Any $o(1/\epsilon)$-round $\ampc$ algorithm using machines with I/O capacity $S = n^{\eps}$ cannot approximate \textsc{2-SLC} within a factor of $1.84-\delta$ for $d = \Omega(\log n/\delta^2)$ under $\ell_2^d$, or within a factor of $3$ for $d = \Omega(n)$ under $\ell_0^d$ or $\ell_1^d$.
\end{theorem}

\subsubsection{Our approach and techniques}
Recall that our first contribution is in extending the lower bound framework of Roughgarden et al.~\cite{Roughgarden} to reason about {\sl adaptive} $\mpc$ computations and {\sl promise problems}.  Each of these poses its own challenges: 

\begin{itemize}[leftmargin=0.5cm] 

\item[$\circ$] {\sl Adaptivity:} In the (non-adaptive) $\mpc$ model, in each round a machine $M$ receives all $S$ bits of its input at once and computes a function of these bits. In other words, each output bit of $M$ is a function of $S$ many input bits it receives. In the adaptive ($\ampc$) setting, on the other hand, recall that $M$ queries for input messages sequentially from the shared memory, where later queries may depend on the outcomes of earlier queries. Thus, each output message of $M$ is the result of a decision tree of depth $S$, which can depend on the contents in as many as $2^S$ locations in the shared memory.

\item[$\circ$] {\sl Promise problems:}  Recalling our discussion of the~\cite{Roughgarden} framework on~\cpageref{RVW-framework}, the key structural lemma connecting $\mpc$ computation to Boolean function complexity that they prove is that every function $g: \zo^N \to \zo$ computable by $R$-round algorithms with machines with I/O capacity $S$ can be represented as polynomials of degree $S^R$.  This reduces the task of proving $\mpc$ lower bounds to that of proving lower bounds on the degree of $g$'s polynomial representation (\Cref{fact:connectivity-degree}).  

In the case of promise problems, one now has to prove a lower bound on the polynomial degree of {\sl partial} Boolean functions: Given a partial Boolean function $g : \Delta \to \zo$ with domain $\Delta \subseteq \zo^N$, one has to show that {\sl every} total function $f: \zo^N\to\zo$ that extends $g$ has to have large degree.  This is challenging since there could be many such extensions~$f$.  This is especially so in the case of the $\OCTC$ problem, since the number of promise inputs (graphs that are a single cycle of length $n$ or two cycles of length $n/2$) is a tiny fraction of all possible inputs (all possible graphs): $n^{O(n)}$ out of $2^{n \choose 2}$.

\end{itemize}

Our solution to incorporating adaptivity is fairly straightforward.  We first observe that the crux of \cite{Roughgarden}'s polynomial method is the fact that a boolean function on $S$ variables can be represented by a polynomial of degree at most $S$. For the adaptive case, we use a generalization of this fact: a decision trees of depth $S$---a strictly larger class than functions on $S$ variables---can also be represented by a polynomial of degree at most $S$.

The promise aspect turns out to pose more of a technical challenge.  Here we depart from the approach of~\cite{Roughgarden} and do not directly prove a lower bound on the polynomial degree (or approximate degree) of any total function $f : \zo^{{n \choose 2}} \to \zo$ that extends $\OCTC$.  Instead, we first reason about other complexity measures of $\OCTC$---its deterministic query complexity in the case of deterministic computation, and approximate certificate complexity in the case of randomized computation---and then leverage classical results from query complexity that relate these measures to polynomial degree.

We defer the precise definitions of these measures to Sections \ref{DetermBounds} and \ref{RandomBounds}.  Roughly speaking, the deterministic query complexity of a partial function $g : \Delta \to \zo$ is the minimum number of queries necessary to distinguish between a $1$-input versus a $0$-input, and its approximate certificate complexity is the minimum certificate complexity of any $g'$ defined on $\Delta$ that closely approximates $g$.  As alluded to above, lower bounding these complexity measures for partial Boolean functions is significantly more challenging than for total Boolean functions: in the case of query complexity, the query algorithm can behave arbitrarily on non-promise instances; in the case of certificate complexity, certificates only have to be valid on promise instances.

We develop new machinery to study these complexity measures. For the $\OCTC$ problem, we prove:

\begin{theorem}[Deterministic query complexity of $\OCTC$] 
\label{thm:OCTC-DT} 
The deterministic query complexity of $\OCTC$ is $\Omega(n^2)$.
\end{theorem}

\begin{theorem}[Approximate certificate complexity of $\OCTC$]
\label{thm:OCTC-cert} 
The $\frac1{6}$-approximate certificate complexity of $\OCTC$ is $\Omega(n)$.
\end{theorem}

\Cref{thm:OCTC-DT,thm:OCTC-cert} are both asymptotically optimal.  Our proof of~\Cref{thm:OCTC-DT} is based on a delicate graph-theoretic adversary argument.  We derive \Cref{thm:OCTC-cert} as a corollary of a general framework for proving lower bounds on approximate certificate complexity:

\begin{theorem}[Framework for proving lower bounds on approximate certificate complexity, see~\Cref{ApproxCertifLB} for a formal statement]
\label{thm:framework} 
Let $g$ be a partial Boolean function and $\delta \in (0,1/2)$. Suppose we can find at least $2K$ pairwise disjoint sensitive blocks of $g$ on each 1-instance such that for each 0-instance $y$, not too many out of those sensitive blocks on all 1-instances are also sensitive blocks of $g$ on $y$. Then the $\delta$-approximate certificate complexity of $g$ is at least $K$.
\end{theorem}

\Cref{thm:OCTC-DT,thm:OCTC-cert} are the query complexity lower bounds that underlie our deterministic and randomized $\ampc$ lower bounds (and can be thought of as being analogous to~\cite{Roughgarden}'s degree lower bounds,~\Cref{fact:connectivity-degree}).  Given~\Cref{thm:OCTC-DT}, one may have expected that our randomized $\ampc$ lower bound would be built on a lower bound on the randomized query complexity of $\OCTC$:   

\begin{conjecture}[Randomized query complexity of $\OCTC$]
\label{conj:randomized}
The randomized query complexity of $\OCTC$ is $\Omega(n^2)$. 
\end{conjecture}

We were unable to resolve~\Cref{conj:randomized}, and as it turns out,~\Cref{thm:OCTC-cert} suffices for our purposes.  However, we still find~\Cref{conj:randomized} to be a natural and independently interesting question.  More generally, we hope that the techniques we have developed to reason about the complexity of {\sl partial} Boolean functions (such as~\Cref{thm:framework}) and the problems that our work leaves open (such as~\Cref{conj:randomized}) will be of independent interest and utility beyond the connection to $\ampc$ lower bounds that originally motivated our work in this direction.






\section{The $\ampc$ model}
In this section, we give a detailed description of the $\ampc$ model that we will work with. Recall that as defined in Behnezhad et al. \cite{Behnezhad}, the $\ampc$ model is centered around the notion of \emph{adaptivity}, which allows machines to communicate using intermediate shared memory called \emph{distributed data stores} (DDS). Between rounds $r$ and $r+1$, machines in round $r$ write their output in terms of key-value pairs to the DDS $\D_r$, and machines in round $r+1$ obtain their input by querying for keys to $\D_r$. Later queries that a machine makes may depend on the keys and responses of earlier queries made by the same machine in the same round. $\D_0$ stores the input.

\pparagraph{Our specifications.} Now we start to describe our specifications of $\ampc$, which differ slightly from those in \cite{Behnezhad} for technical reasons. In each round, in the DDS, all values under the same key are stored as a multiset under the key. This means that duplicate values are allowed, and each value is written by a unique machine. When a machine queries for a key to the DDS, the response is the entire multiset stored under the key; if there are no values stored under the key, the response is empty (in the form of an empty multiset).

Let $S$ be the I/O capacity. For each machine in each round, we require that the sum of the total number of values in all responses and the number of queries with an empty response is at most $S$, and the machine writes at most $S$ key-value pairs to the DDS. We note that this is essentially the same constraint as in \cite{Behnezhad}. In addition, we require that in each round, under any key, there are at most $S$ values (including duplicates) written to the DDS.

We consider $\ampc$ algorithms that compute a total or partial Boolean function on $N$ input bits, that is, a function $g:\Delta \to \{0,1\}$ with domain $\Delta \subseteq \{0,1\}^N$. The input is stored in $\D_0$ in terms of $N$ key-value pairs $(i, x_i)$, where each $x_i$ is the $i$-th bit of the input. After the final round, the DDS contains a single key-value pair $(\textsc{answer}, 1)$ or $(\textsc{answer},0)$ that indicates the final computation result. We say an $\ampc$ algorithm $\A$ computes the function $g$ in $R$ rounds if during $\A$'s computation on any input $x \in \Delta$, $\D_R$ contains a single key-value pair $(\textsc{answer},g(x))$.

\pparagraph{Graph problems in $\ampc$.} We will represent undirected graph problems on $n$ vertices as total or partial Boolean functions on $N = {n \choose 2}$ input bits (and represent directed graph problems as Boolean functions on $N = 2{n \choose 2}$ input bits), each of which indicates whether the corresponding edge is present or not. In other words, we represent an input graph using its adjacency matrix. For instance, the $\OCTC$ problem on $n$ vertices (with $n$ even) can be represented by a partial Boolean function $\OCTC:\Delta_{\OCTC} \to \{0,1\}$ on $N = {n \choose 2}$ input bits, where $\Delta_{\OCTC} \subset \{0,1\}^N$ is the set of all 1-cycle instances and 2-cycle instances. We will discuss this function in more detail in subsequent sections.

\begin{remark}[Efficiently converting adjacency-matrix representations of graphs into adjacency-list representations] \rm{
In this remark we observe that the upper bounds of \cite{Behnezhad} hold if the input graph is represented by its adjacency matrix (which as discussed above, is the representation that we work with throughout this paper). In \cite{Behnezhad}, the input graph is specified by a list of its edges; to be explicit, if $m$ is the number of edges in the graph, then $\D_0$ stores $m$ key-value pairs $(i, (u_i, v_i))$, where $u_i$ and $v_i$ are the endpoints of the $i$-th edge. As in \cite{Behnezhad}, take $S = n^\epsilon$ for $\epsilon \in (0,1)$. We show that starting with an input graph represented by its adjacency matrix, we can use $O(1/\epsilon)$ rounds of preprocessing to rewrite it in the format of \cite{Behnezhad}. In particular, we need to count the number $m$ of present edges and give a labeling of these edges from 1 through $m$. To do this, we first partition the $N$ input bits into $P = O(N/S)$ groups $X_1, \dots, X_P$ each of size at most $S-2$, and take $P$ machines $M_1, \dots, M_P$. In one round, each $M_i$ reads the input bits in the group $X_i$ and outputs the number $a_i$ of present edges in $X_i$ as well as a list $E_i$ of these $a_i$ edges. Then in $O(1/\epsilon)$ rounds, we compute the prefix sum $c_1, \dots, c_P$ of the numbers $a_1, \dots, a_P$; this is possible because computing the prefix sum of a sequence in the $\mpc$ model can be done in $O(1/\epsilon)$ rounds \cite{Goodrich}, and any $\mpc$ algorithm can be simulated by an $\ampc$ algorithm with the same number of rounds. Note that $m = c_P$. In this process, we use a separate collection of machines to preserve the lists $E_i$'s in the DDS. Finally, in one round, we let each $M_i$ read $c_{i-1}$ (set $c_0 = 0$), $c_i$, and the list $E_i$, and label edges in $E_i$ from $c_{i-1}+1$ through $c_i$.
}\end{remark}

\pparagraph{Processing invalid inputs.} For an input $x \in \{0,1\}^N$, we say $x$ is \emph{valid} if $x \in \Delta$ and \emph{invalid} otherwise. We assume that an $R$-round $\ampc$ algorithm $\A$ that computes $g$ can also process an invalid input $x$, in the following sense: Machines read and process the bits of $x$ in rounds as if a valid input were being computed. In a round of $\A$'s computation on $x$, it may happen that after a machine $v$ makes some queries and receives the responses, the sum of the total number of values in the responses and the number of queries with an empty response exceeds $S$. In this case, we may let $v$ stop making further queries and write nothing to the DDS. An important assumption that we make is that in any round of $\A$'s computation on any \emph{invalid} input $x$, there are also at most $S$ values written under any single key to the DDS. After round $R$, $\D_R$ stores a multiset of bits under the key $\textsc{answer}$; this multiset has at most $S$ bits, and can be empty.

\begin{remark}[Controlling the behavior of \cite{Behnezhad}'s $\OCTC$ algorithm on invalid inputs]\rm{We note that the $\ampc$ algorithm for $\OCTC$ given by \cite{Behnezhad} can be easily modified to satisfy all our technical restrictions described in this section, specifically those on processing \emph{invalid} inputs. For such a modification, we need to resolve the following issue: the algorithm of \cite{Behnezhad} needs to query for neighborhoods of the vertices, which is problematic on invalid inputs where some vertex has degree greater than $S$. Such invalid inputs can be identified by $O(1/\epsilon)$-rounds of preprocessing. Specifically, we may compute the degree of all vertices in $O(1/\epsilon)$ rounds; this is possible in the $\mpc$ model \cite{Behnezhad2}, and any $\mpc$ algorithm can be simulated by an $\ampc$ algorithm with the same number of rounds. If any vertex has degree exceeding $S$, we generate an error message that halts the computation.
}\end{remark}

\section{Generalizing \cite{Roughgarden}'s Polynomial Method: Adaptivity and Partial Boolean Functions}
To prove lower bounds on the round complexities of Boolean functions in the $\mpc$ model, Roughgarden et al. \cite{Roughgarden} introduce a variant of the ``polynomial method,'' which shows that a function computable by an efficient deterministic $\mpc$ algorithm can be represented by a polynomial with low degree. In this section, we generalize the polynomial representation construction of \cite{Roughgarden} to the $\ampc$ model. In addition, we generalize this construction to \emph{partial} Boolean functions, in the sense that given a partial Boolean function computable by an efficient deterministic $\ampc$ algorithm, we show that it can be extended to a \emph{total} Boolean function that can be represented by a polynomial with low degree. This generalized construction is central to our analysis of round complexities of partial Boolean functions in $\ampc$ in subsequent sections.

\subsection{Efficient Deterministic $\ampc$ Algorithms Imply Small Polynomial Degree}
For a Boolean function $g:\{0,1\}^N \to \{0,1\}$, the \emph{degree} of $g$, denoted $\deg(g)$, is the degree of the unique multilinear polynomial $p(x_1, \dots, x_N)$ such that $p(x) = g(x)$ for all $x \in \{0,1\}^N$ (see e.g. \cite{O'Donnell}). We recall the following result from \cite{Roughgarden}: 

\begin{theorem}[\cite{Roughgarden}, Theorems 3.1 and 3.7] \label{TimShuffle}
If a Boolean function $g:\{0,1\}^N \to \{0,1\}$ can be computed by an $R$-round deterministic $\mpc$ algorithm, then $\deg(g) \le S^R$.
\end{theorem}

We now generalize Theorem \ref{TimShuffle} to incorporate adaptivity and promise problems: 

\begin{theorem} \label{AMPCDeg}
Let $g: \Delta \to \{0,1\}$ be a partial Boolean function with domain $\Delta \subseteq \{0,1\}^N$. If $g$ can be computed by an $R$-round deterministic $\ampc$ algorithm, then there exists a polynomial $p(x_1, \dots, x_N)$ with degree at most $S^{2R}$ such that $p(x) = g(x)$ for any $x \in \Delta$ and $p(x) \in \{0,1\}$ for any $x \in \{0,1\}^N \setminus \Delta$. In particular, if $g$ is a total Boolean function, then $\deg(g) \le S^{2R}$.
\end{theorem}

\begin{proof}
Let $\A$ be an $R$-round deterministic $\ampc$ algorithm that computes $g$. As we mentioned in the previous section, we assume $\A$ can process invalid inputs as well. We start by setting up some notations. For round $1 \le r \le R+1$ and machine $v$, the sequence of queries made by $v$ in round $r$ of $\A$'s computation on an input and their responses is specified by a sequence of at most $S$ pairs of keys and multisets. We call this sequence the \emph{query sequence} of $v$ in round $r$. We denote $\Sigma_{r,v}$ as the set of possible query sequences of $v$ in round $r$ during $\A$'s computation on some input in $\{0,1\}^N$. For round $0 \le r \le R$ and key $k$, we denote $\Gamma_{r,k}$ as the set of possible multisets of values written to $\D_r$ under $k$ during $\A$'s computation on some input in $\{0,1\}^N$. Since we assumed that during $\A$'s computation on any input, there are at most $S$ values written to $\D_r$ under $k$, each multiset in $\Gamma_{r,k}$ has size at most $S$.

To construct the desired polynomial $p$, we will construct two families of polynomials:
\begin{itemize}
    \item[$\circ$] $p_{r,v,z}(x_1, \dots, x_N)$ for each round $1 \le r \le R$, machine $v$, and query sequence $z \in \Sigma_{r,v}$, which has degree at most $S^{2r-1}$ and satisfies that for any $x \in \{0,1\}^N$,  $p_{r,v,z}(x)=1$ if in round $r$ of $\A$'s computation on $x$, the query sequence of $v$ is $z$, and $p_{r,v,z}(x)=0$ otherwise.
    
    \item[$\circ$] $q_{r,k,W}(x_1, \dots, x_N)$ for each round $0 \le r \le R$, key $k$, and multiset $W\in \Gamma_{r,k}$, which has degree at most $S^{2r}$ and satisfies that for any $x \in \{0,1\}^N$, $q_{r,k,W}(x)=1$ if the values stored under $k$ in $\D_r$ during $\A$'s computation on $x$ is $W$, and $q_{r,k,W}(x)=0$ otherwise.
\end{itemize}
The construction will proceed by induction on $r$ as follows: As the base case we construct the polynomials $q_{0,k,W}$'s. Then inductively for each $r \ge 1$, we first construct the polynomials $p_{r, v,z}$'s from the $q_{r-1,v,z}$'s, and then construct the polynomials $q_{r,k,W}$'s from the $p_{r,v,z}$'s.

We start with the base case $r = 0$. The initial DDS $\D_0$ stores the input in terms of $N$ key-value pairs $(i, x_i)$, where $x_i$ is the $i$-th input bit. The possible keys that appear in $\D_0$ are the indices $1 \le i \le N$, and the multiset of values stored under any key is either $\{0\}$ or $\{1\}$. For each $i$, we take
\[  q_{0,i,\{0\}}(x_1, \dots, x_N) = 1 - x_i, \qquad q_{0,i,\{1\}}(x_1, \dots, x_N) = x_i, \]
so that for all $x \in \{0,1\}^N$, $q_{0,i,\{b\}}(x) = 1$ if and only if the $i$-th bit of $x$ is $b$. Note that each $q_{0,k,W}$ has degree $1 = S^0$.

Now we begin the inductive step $r \ge 1$. We first construct $p_{r,v,z}$, where $v$ is a machine and $z \in \Sigma_{r,v}$ is a query sequence. Write $z = ((k_1,W_1), \dots, (k_s,W_s))$, where $s \le S$ and the $j$-th query has key $k_j$ and response $W_j$. By adaptivity, we see that the query sequence of $v$ in round $r$ is $z$ if and only if for each $j$, the values under $k_j$ in $\D_{r-1}$ is exactly $W_j$. Therefore, we take
\[  p_{r, v, z}(x_1, \dots, x_N) = \prod_{j = 1}^s q_{r-1,k_j,W_j}(x_1, \dots, x_N),   \]
which restricts to a Boolean function on $\{0,1\}^N$. Since each $q_{r-1, k_j, W_j}$ has degree at most $S^{2r-2}$ and $s \le S$, the degree of $p_{r,v,z}$ is at most $S^{2r-1}$.

Next we construct the polynomials $q_{r,k,W}$'s. We introduce some notations. Denote $M$ as the set of all machines involved in $\A$. For each round $r$, machine $v$, key $k$, and nonempty multiset $W$, let $\Sigma_{r,v,k,W} \subseteq \Sigma_{r,v}$ be the subset of query sequences of $v$ in round $r$ that cause $v$ to write exactly the values in $W$ to $\D_r$ under $k$. Note that for fixed $v, r, k$, the sets $\Sigma_{r-1,v,k,W}$ are mutually disjoint.

From now on, we fix a key $k$. Given a multiset $W \in \Gamma_{r,k}$, we will look at all possible assignments of values in $W$ to the machines in round $r$ and, for each assignment, construct a polynomial that indicates the event that any machine that is assigned some values writes exactly the assigned values under $k$ to $\D_r$. For an assignment $\alpha: W \to M$, consider the polynomial
\begin{equation}\label{qprepare}
\prod_{v \in \alpha(W)} \sum_{z \in \Sigma_{r,v, k, \alpha^{-1}(v)}} p_{r,v,z}.
\end{equation}
For each input $x \in \{0,1\}^N$, at most one sequence $z \in \Sigma_{r,v}$ can be the query sequence of $v$ in round $r$ of $\A$'s computation on $x$. Thus in the summation in (\ref{qprepare}), at most one $p_{r, v, z}$ evaluates to 1 on $x$. Hence the polynomial (\ref{qprepare}) restricts to a Boolean function on $\{0,1\}^N$. It is easy to verify that (\ref{qprepare}) evaluates to 1 on an input $x \in \{0,1\}^N$ if and only if during $\A$'s computation on $x$, any machine $v$ in the image of $\alpha$ writes precisely the values $\alpha^{-1}(v)$ under the key $k$ to $\D_r$. Moreover, since no restriction is placed on what the machines outside $\alpha(W)$ write to $\D_r$ under $k$, an alternative way to interpret (\ref{qprepare}) is that it evaluates to 1 on $x$ if and only if during $\A$'s computation on $x$, the multiset of values written to $\D_r$ under $k$ is a \emph{superset} of $W$ and the sources of the values in $W$ are indicated by the assignment $\alpha$.

Taking into account all possible assignments, we take
\begin{equation}\label{tildeq}
\tilde{q}_{r,k,W} = \sum_{\alpha: W \to M} \prod_{v \in \alpha(W)} \sum_{z \in \Sigma_{r, v, k, \alpha^{-1}(v)}} p_{r, v,z}.    
\end{equation}
During $\A$'s computation on an input $x \in \{0,1\}^N$, if the multiset of values under the key $k$ in $\D_r$ is a superset of $W$, then the sources of the values in $W$ are indicated by a unique assignment $\alpha:W \to M$. Thus at most one summand in the outer summation in (\ref{tildeq}) evaluates to 1 on $x$. This verifies that $\tilde{q}_{r,k,W}$ restricts to a Boolean function on $\{0,1\}^N$. Moreover, $\tilde{q}_{r,k,W}$ evaluates to 1 on an input $x \in \{0,1\}^N$ if and only if during $\A$'s computation on $x$, the multiset of values written to $\D_r$ under $k$ is a superset of $W$. Since $W$ has size at most $S$, $\alpha(W)$ has size at most $S$ as well. Then since each $p_{r, v,z}$ has degree at most $S^{2r-1}$, we have that $\tilde{q}_{r,k,W}$ has degree at most $S^{2r}$.

Now we use the polynomials $\tilde{q}_{r,k,W}$'s to construct our desired polynomials $q_{r,k,W}$'s for $W \in \Gamma_{r,k}$. We use a downward induction on the size of $W$. For a multiset $W \in \Gamma_{r,k}$ with maximum possible size $S$, we can simply take
\[  q_{r,k,W} = \tilde{q}_{r,k,W}.  \]
For a multiset $W \in \Gamma_{r,k}$ with $|W| < S$, to obtain $q_{r,k,W}$ from $\tilde{q}_{r,k,W}$, we need to rule out inputs on which the values written to $\D_r$ under $k$ is a strict superset of $W$. Thus we take
\begin{equation}\label{qequation}
q_{r,k,W} = \tilde{q}_{r,k,W} - \sum_{W' \in \Gamma_{r,k}, W' \supsetneq W} q_{r,k,W'}.
\end{equation}  
For $x \in \{0,1\}^N$, if $\tilde{q}_{r,k,W}(x) = 1$, then at most one $q_{r,k,W'}$ in the summation in (\ref{qequation}) evaluates to 1 on $x$ since the values stored under $k$ in $\D_r$ during $\A$'s computation on $x$ is a unique multiset; if $\tilde{q} _{r,k,W}(x) = 0$, then multiset of values stored under $k$ in $\D_r$ during $\A$'s computation on $x$ is not a superset of $W$, and all the $q_{r,k,W'}$'s in the summation in (\ref{qequation}) evaluates to 0 on $x$. Thus $q_{r,k,W}$ restricts to a Boolean function on $\{0,1\}^N$. It is easy to see that $q_{r,k,W}$ has degree at most $S^{2r}$. This concludes our inductive construction.

To finish the proof, we take $p = q_{R, \textsc{answer}, \{1\}}$. It is straightforward to check that $p$ has degree at most $S^{2R}$, restricts to a Boolean function on $\{0,1\}^N$, and, among the valid inputs, takes value 1 precisely on those on which $g$ evaluates to 1.
\end{proof}

\subsection{Implications for the Randomized Setting}
In the randomized setting, Roughgarden et al \cite{Roughgarden} used Theorem \ref{TimShuffle} to show that a total Boolean function that can be computed by a small-round randomized $\mpc$ algorithm has small approximate degree (see Theorem 3.5 of \cite{Roughgarden}). With Theorem \ref{AMPCDeg}, we are able to give an analog in the $\ampc$ model for total Boolean functions using the same proof as that of Theorem 3.5 in \cite{Roughgarden}.

\begin{definition}[Randomized $\ampc$ algorithm]\label{RandomAMPC}
A \emph{randomized $\ampc$ algorithm} is a probability distribution over deterministic $\ampc$ algorithms. The number of rounds required by a randomized $\ampc$ algorithm is the maximum number of rounds required by a deterministic $\ampc$ algorithm in the support of the distribution. For $\delta \in [0,1)$, we say that a randomized $\ampc$ algorithm $\A$ computes a partial Boolean function $g:\Delta \to \{0,1\}$ with error at most $\delta$ if for each $x \in \Delta$, $\A$ outputs $g(x)$ with probability at least $1-\delta$.
\end{definition}

We say a \emph{total} Boolean function $g:\{0,1\}^N \to \{0,1\}$ is \emph{approximately represented} by a polynomial $p$ if $|p(x) - g(x)| \le \frac{1}{3}$ for any $x \in \{0,1\}^N$. The \emph{approximate degree} of $g$, denoted $\widetilde\deg(g)$, is given by
\[ \min \{\deg(p) \mid p \mbox{ approximately represents } g\}.    \]

\begin{theorem}\label{AMPCApproxDeg}
If a total Boolean function $g:\{0,1\}^N \to \{0,1\}$ can be computed by an $R$-round randomized $\ampc$ algorithm with error at most $1/3$, then $g$ has approximate degree at most $S^{2R}$.
\end{theorem}

\begin{proof}
Let $\A$ be an $R$-round randomized $\ampc$ algorithm that computes $g$ with error at most $1/3$. Then $\A$ is a distribution over $R$-round deterministic $\ampc$ algorithms $\A_i$'s, where each $\A_i$ has weight $w_i$. For each $i$, by Theorem \ref{AMPCDeg}, there is a polynomial $p_i$ of degree at most $S^{2R}$ such that $\A_i$ outputs $p_i(x)$ for all $x \in \{0,1\}^N$. Take $p = \sum_{i} w_i p_i$, which is a polynomial of degree at most $S^{2R}$. Then for all $x \in \{0,1\}^N$, $p(x)$ equals the probability that $\A$ outputs 1 on $x$. Since $\A$ computes $g$ with error at most $1/3$, it is easy to see that $p$ approximately represents $g$.
\end{proof}

\begin{remark}\rm{
Note that Theorem \ref{AMPCApproxDeg} applies only to total Boolean functions. An issue with generalizing this result to a partial Boolean $g: \Delta \to \{0,1\}$ is as follows. Note that each $p_i$ as in the proof of Theorem \ref{AMPCApproxDeg} is constructed using Theorem \ref{AMPCDeg}. The statement of Theorem \ref{AMPCDeg} leaves the possibility that for some $\A_i, \A_j$ in the distribution of $\A$, $p_i(x) = 0$ and $p_j(x) = 1$ for some invalid input $x \in \{0,1\}^N \setminus \Delta$. Then some family of weights $\{w_i\}$ may result in $p(x)\in (1/3, 2/3)$, which means that $p$ cannot approximately represent any total Boolean function. This issue goes away if, for instance, in Theorem \ref{AMPCDeg}, we can guarantee in addition that the polynomial $p$ evaluates to 0 on all invalid inputs. However, this stronger version is out of our reach.
}\end{remark}

\section{Deterministic Round Lower Bound for $\OCTC$ via Query Complexity} \label{DetermBounds}
For a total Boolean function, Theorems \ref{AMPCDeg} and \ref{AMPCApproxDeg} from the previous section imply the following deterministic and randomized round lower bounds in the $\ampc$ model in terms of the degree and approximate degree of the function respectively:

\begin{corollary}
Any deterministic $\ampc$ algorithm that computes a total Boolean function $g:\{0,1\}^N \to \{0,1\}$ requires $\frac{1}{2} \log_S \deg(g)$ rounds. Any randomized $\ampc$ algorithm that computes $g$ with error at most 1/3 requires $\frac{1}{2} \log_S \widetilde{\deg}(g)$ rounds.
\end{corollary}

For a general partial Boolean function, its degree and approximate degree are less well-understood.  Nonetheless, we can use Theorem \ref{AMPCDeg} to relate the number of rounds required by a deterministic $\ampc$ algorithm to compute a partial Boolean function to other complexity measures of the function, thus obtaining deterministic round lower bounds. In this section, we elaborate on the relation between the number of rounds and the \emph{deterministic query complexity} of the function. We will in particular focus on the partial Boolean function $\OCTC: \Delta_{\OCTC} \to \{0,1\}$ . We will prove an asymptotically optimal lower bound on the deterministic query complexity of $\OCTC$, and then use it to obtain an $\Omega(\log_S n)$ deterministic round lower bound for computing $\OCTC$ in the $\ampc$ model. A randomized round lower bound will be deferred to Section \ref{RandomBounds}.

\subsection{Efficient Deterministic $\ampc$ Algorithms Imply Small Deterministic Query Complexity}
We start by defining the deterministic query complexity of a partial Boolean function.
\begin{definition}[Deterministic query complexity]\label{DetermQuery}
For a partial Boolean function $g : \Delta \to \zo$ with domain $\Delta \subseteq \zo^N$, we define its \emph{deterministic query complexity}, or commonly also referred to as deterministic decision tree complexity, to be 
\[ D(g) \coloneqq \min\{ D(f)\, \mid\, f : \zo^N \to \zo \text{ where $f(x) = g(x)$ for all $x \in \Delta$}. \} \] 
\end{definition}

The current best-known relation between the degree and the deterministic query complexity of a total Boolean function is the following:
\begin{lemma}[Midrij\=anis \cite{Midrijanis}]\label{polydegreeDT}
For a Boolean function $g:\{0,1\}^N \to \{0,1\}$, we have
\[  D(g) \le 2\deg(g)^3.    \]
\end{lemma}

This relation allows us to obtain the following useful corollary of Theorem \ref{AMPCDeg}, which reduces proving deterministic round lower bounds for computing a partial Boolean function $g$ in the $\ampc$ model to proving lower bounds on $D(g)$:

\begin{theorem}\label{AMPCDT}
For any $R$-round deterministic $\ampc$ algorithm that computes a partial Boolean function $g: \Delta \to \{0,1\}$, we have
\[  D(g) \le 2S^{6R}. \]
In particular, $R \ge \frac{1}{6} \log_S \frac{D(g)}{2}$.
\end{theorem}

\begin{proof}
Let $\A$ be an $R$-round deterministic $\ampc$ algorithm that computes $g$. By Theorem \ref{AMPCDeg}, there exists a polynomial $p(x_1, \dots, x_N)$ with degree at most $S^{2R}$ such that $p(x) = g(x)$ for any $x \in \Delta$ and $p(x) \in \{0,1\}$ for any $x \in \{0,1\}^N \setminus \Delta$. We denote the Boolean function obtained by restricting $p$ to $\{0,1\}^N$ as $\tilde g$ (which is equal to $g$ if $g$ is a total Boolean function). Then $\deg(\tilde g) \le S^{2R}$. Lemma \ref{polydegreeDT} then implies that
\[     D(\tilde g) \le 2\deg(\tilde g)^3 \le 2S^{6R}. \]
Since $g$ is the restriction of $\tilde g$ on $\Delta$, we have $D(g) \le D(\tilde g)$. The theorem thus follows.
\end{proof}

\subsection{Deterministic Round Complexity of $\OCTC$}
We now focus on the $\OCTC$ problem and use Theorem \ref{AMPCDT} to obtain a lower bound on its deterministic round complexity in $\ampc$. This requires us to lower bound $D(\OCTC)$. We will actually establish the following lower bound on the deterministic query complexity of the more general promise problem $\OCKC$ of distinguishing between a cycle of length $n$ versus $k$ cycles of length $\frac{n}{k}$, where $k$ divides $n$. The proof is deferred to Section \ref{1CkCDTproof}.

\begin{theorem}[Deterministic query complexity of $\OCKC$]\label{1CkCDT}
\[  D(\OCKC) \ge \frac{n^2}{128k^2}.    \]
\end{theorem}

In particular, for $k = 2$, we have:
\begin{corollary}[Deterministic query complexity of $\OCTC$]\label{1C2CDT}
\[  D(\OCTC) \ge \frac{n^2}{512}.   \]
\end{corollary}

We are now ready to give a deterministic round lower bound for computing $\OCTC$ in $\ampc$.

\begin{theorem} \label{1C2Cdeterministic}
Any deterministic $\ampc$ algorithm that computes $\OCTC$ requires $\frac{1}{3}\log_S n - \frac{1}{3}\log_S 32 = \Omega(\log_S n)$ rounds. In particular, if $S = n^\epsilon$ for $\epsilon \in (0,1)$, then any such deterministic $\ampc$ algorithm requires $\Omega(1/\epsilon)$ rounds.
\end{theorem}

\begin{proof}
Let $\A$ be an $R$-round deterministic $\ampc$ algorithm that computes $\OCTC$. By Theorem \ref{AMPCDT} and Corollary \ref{1C2CDT}, we have
\[  R \ge \frac{1}{6} \log_S \frac{D(\OCTC)}{2} \ge \frac{1}{6} \log_S \frac{n^2}{1024} = \frac{1}{3}\log_S n - \frac{1}{3}\log_S 32.  \qedhere \] 
\end{proof}

\subsection{Proof of Theorem \ref{1CkCDT}}\label{1CkCDTproof}
In this subsection, we prove the lower bound on the deterministic query complexity of the $\OCKC$ problem given in Theorem \ref{1CkCDT}. 

We consider the following adversary strategy. The adversary maintains two graphs $Y$ and $M$. The graph $Y$ contains the \emph{$\yes$-edges}, i.e. edges for which the adversary has replied $\yes$. The graph $M$, which stands for ``maybe,'' contains edges for which the adversary has not yet replied $\no$. We call the edges that are not in $M$ \emph{$\no$-edges}. The graph $M$ has $n$ vertices at all times and is initially a clique. The graph $Y$ is a subgraph of $M$ at all times and is initially empty.

The idea is that the adversary gives away an edge (i.e. replies $\yes$) once it has an endpoint that is incident to sufficiently many (but not too many) edges for which the adversary has replied $\no$. In this way, we will have enough $\no$ queries after a small number of $\yes$ queries are made, and at that point, the degrees of all vertices in $M \setminus Y$ are still high enough to ensure that $M$ contains both a Hamiltonian cycle and $k$ disjoint cycles of length $\frac{n}{k}$. Procedure 1 below formally describes the way the adversary processes an edge query $(u,v)$ in $M$.

\begin{algorithm}[htb]
\caption{\textsc{AdversaryStrategy}($(u,v)$)} \label{adversary}
\begin{algorithmic}[1]
  \If {$|V(Y)| \le \frac{n}{4k}-1$}
    \If {one of $u, v$ is in $Y$ and has degree 2 in $Y$} \Comment{Step 1}
        \State Remove $(u,v)$ from $M$ and reply $\no$ 
    \ElsIf {both $u,v$ are in $Y$} \Comment{Step 2}
        \State Remove $(u,v)$ from $M$ and reply $\no$
    \ElsIf {the adversary has replied $\no$ to fewer than $\frac{n}{4k}$ edges that are incident to $u$ \newline \hspace*{48.5pt} \emph{and} fewer than $\frac{n}{4k}$ edges that are incident to $v$} \Comment{Step 3}
        \State Remove $(u,v)$ from $M$ and reply $\no$
    \Else \Comment{Step 4}
        \State Add $(u,v)$ to $Y$ and reply $\yes$  
    \EndIf
  \Else \label{changephase}
    \State Choose some 1-cycle or $k$-cycle configuration that is consistent with $M$ and $Y$
    \State Update $M$ and $Y$ and reply accordingly
  \EndIf
\end{algorithmic}
\end{algorithm}

We will divide the game between the adversary and the algorithm into two phases. We say that the game is in \emph{Phase 1} if the adversary has never hit line \ref{changephase} while running Procedure \ref{adversary} on the queries, and is in \emph{Phase 2} otherwise. In other words, when the adversary first hits line \ref{changephase}, Phase 1 ends and Phase 2 starts. Observe that we have $|V(Y)| \le \frac{n}{4k}+1$ at anytime in Phase 1, and $|V(Y)| \ge \frac{n}{4k}$ at the end of Phase 1. We will show later that when Phase 1 ends, the graph $M$ contains a Hamiltonian cycle that contains $Y$, as well as $k$ disjoint cycles of length $\frac{n}{k}$ whose union contains $Y$. In other words, the answer remains ambiguous to the algorithm throughout Phase 1. Then at the beginning of Phase 2, the adversary can simply raise the white flag and choose a particular 1-cycle or $k$-cycle configuration to stick to throughout the rest of the game. However, before that, we first show that the above strategy gives us the desired lower bound on $D(\OCKC)$.

\begin{claim}
By the end of Phase 1, the adversary has replied $\no$ to $\frac{n^2}{128k^2}$ queries.
\end{claim}

\begin{proof}
We count the $\no$-edges that has at least one endpoint in $V(Y)$ at the end of Phase 1. For each edge $e \in E(Y)$, by Step 3 of the adversary strategy, at least one of the two endpoints of $e$ is incident to at least $\frac{n}{4k}$ $\no$-edges. Moreover, by Step 1, each vertex in $V(Y)$ has degree at most 2. Thus, there are at least $\frac{|E(Y)|}{2}$ vertices in $V(Y)$ each of which is incident to at least $\frac{n}{4k}$ $\no$-edges. Therefore, the number of $\no$-edges that has at least one endpoint in $V(Y)$ is at least
\[  \frac{1}{2}  \cdot \frac{n}{4k}\cdot\frac{|E(Y)|}{2}= \frac{n}{16k}|E(Y)|,    \]
where the additional factor of $\frac{1}{2}$ accounts for that a $\no$-edge may be counted at most twice. At the end of Phase 1, $|E(Y)| \ge \frac{1}{2} |V(Y)| \ge \frac{n}{8k}$. Thus the number of $\no$-edges at this point is at least $\frac{n^2}{128k^2}$.
\end{proof}

Through the remaining claims of the subsection, we show that when Phase 1 ends, the graph $M$ contains a Hamiltonian cycle that contains $Y$, as well as $k$ disjoint cycles of length $\frac{n}{k}$ whose union contains $Y$. Denote $c = |\mathcal{C}(Y)|$ as the number of connected components in $Y$.

\begin{claim}\label{paths}
At any time in Phase 1, $Y$ consists of $c$ paths. In particular, $|V(Y)| = |E(Y)| + c$.
\end{claim}

\begin{proof}
In Phase 1, Step 1 of the adversary strategy ensures that no vertex in $Y$ has degree 3 or more, and Step 2 ensures that $Y$ is acyclic.
\end{proof}

Thus, for $i = 1, \dots, c$, we denote the two endpoints of the $i$-th path in $Y$ as $a_i$ and $b_i$. Moreover, we denote $M'$ as the subgraph of $M$ induced by the vertices in $V(M)\setminus V(Y)$. Since $|V(M)| = n$ at all times, $|V(M')| = n - |V(Y)|$.

\begin{claim}\label{Mdegree}
At any time in Phase 1, each $a_i$, $b_i$ has at least $\frac{(4k-1)n}{4k}-|V(Y)|-1$ neighbors in $M$ that are also vertices in $M'$, and each vertex in $M'$ has at least $\frac{(4k-1)n}{4k}-|V(Y)|-1$ neighbors in $M'$.
\end{claim} 

\begin{proof}
Let $v$ be one of the $a_i$'s, $b_i$'s, or a vertex in $M'$. Then, either $v \not \in V(Y)$, or $v$ has degree 1 in $Y$. Suppose $v$ has fewer than $\frac{(4k-1)n}{4k}-|V(Y)|-1$ neighbors in $M$ that are also vertices in $M'$. In other words, the number of vertices in $M'$ that are not adjacent to $v$ in $M$ is greater than
\[  |V(M')| - \left(\frac{(4k-1)n}{4k} - |V(Y)|-1 \right) = n - |V(Y)| - \left(\frac{(4k-1)n}{4k} - |V(Y)|-1 \right) = \frac{n}{4k}+1.    \]
Then we can find distinct vertices $w_1, \dots, w_{\ceil{\frac{n}{4k}}+1} \in V(M')$ that are not adjacent to $v$ in $M$, i.e. the adversary has replied $\no$ to the query $(v, w_j)$ for each $j$. Without loss of generality, let $(v, w_{\ceil{\frac{n}{4k}}+1})$ be the most recent query among the $(v, w_j)$'s. Then, when the query $(v, w_{\ceil{\frac{n}{4k}}+1})$ was examined, the conditions in Steps 1-3 of the adversary strategy all failed. The adversary would then reply $\yes$ to $(v, w_{\ceil{\frac{n}{4k}}+1})$ by Step 4, which is a contradiction.
\end{proof}

The following simple counting lemma will be useful.

\begin{claim}\label{edgecount}
Let $n,m$ be positive integers with $\frac{n}{2} \le m \le n$. Let $H$ be a cycle on $n$ vertices, and $A, B$ be two subsets of vertices each of size at least $m$. Then there are at least $2m-n$ edges $(u,v)$ in $H$ such that $u \in A, v \in B$ or $u \in B, v \in A$. 
\end{claim}

\begin{proof}
An edge $(u,v)$ that does not satisfy the above condition falls into one of the following three categories:
\begin{itemize}
    \item[$\circ$] One of $u,v$ is contained in $(A \cup B)^C$. There are at most $2(n-|A \cup B|)$ such edges.
    \item[$\circ$] Both $u,v$ are contained in $A \setminus (A\cap B)$. There are at most $\max\{|A| - |A \cap B| - 1, 0\}$ such edges.
    \item[$\circ$] Both $u,v$ are contained in $B \setminus (A \cap B)$. There are at most $\max\{|B| - |A \cap B| -1, 0\}$ such edges.
\end{itemize}
Summarizing the three cases above, we see that the number of edges satisfying the condition of the claim is at least
\begin{align*}
    &\ \  \ \  n - (2(n - |A \cup B|) + \max\{|A| - |A \cap B| - 1, 0\} + \max\{|B| - |A \cap B| - 1, 0\})\\
    &\ge  n - (2(n - |A \cup B|) + (|A| - |A \cap B|) + (|B| - |A \cap B|))\\
    &= (|A \cup B| - |A|+|A \cap B|) + (|A \cup B| - |B| + |A \cap B|) -n\\
    &=|B| + |A| -n\\
    &\ge 2m - n. \qedhere
\end{align*}  
\end{proof}

\begin{claim}\label{Hcycle}
At any time in Phase 1, $M$ contains a Hamiltonian cycle that contains $Y$.
\end{claim}

\begin{proof}
By Claim \ref{Mdegree}, in the graph $M'$, the degree of each vertex is at least
\[  \frac{(4k-1)n}{4k} - |V(Y)| -1 \ge \frac{1}{2}(n - |V(Y)|) = \frac{1}{2}|V(M')|,   \]
where we used that $|V(Y)| \le \frac{n}{4k}+1$ in Phase 1 for the inequality\footnote{Strictly speaking, for this estimate and similar ones in the rest of the subsection to hold, we require that $n \ge 28k$. Note that when $n<28k$, Theorem \ref{1C2CDT} holds simply because $\frac{n^2}{128k^2}< \frac{n}{k} = C(\OCKC)\le D(\OCKC)$. (See Example \ref{Certif1CkC}.)}. By Dirac's Theorem, we can find a Hamiltonian cycle $H$ of $M'$. We now construct a Hamiltonian cycle of $M$ by replacing $c$ distinct edges in $H$ by the $c$ paths in $Y$.

For each $i = 1, \dots, c$, by Claim \ref{Mdegree}, each of $a_i, b_i$ is adjacent in $M$ to at least $\frac{(4k-1)n}{4k}-|V(Y)|-1$ vertices in $M'$. Let $P_i$ denote the set of edges in $H$ that has one endpoint adjacent to $a_i$ and the other endpoint adjacent to $b_i$. By Claim \ref{edgecount} applied to the cycle $H$ and the two subsets of vertices that are adjacent to $a_i$, $b_i$ respectively, we see that
\[  |P_i| \ge 2 \bigg(\frac{(4k-1)n}{4k}-|V(Y)| -1 \bigg) - (n- |V(Y)|) = \frac{(2k-1)n}{2k}- |V(Y)| -2 \ge \frac{(4k-3)n}{4k} -3.\]

Note that in Phase 1, we have $c \le \frac{1}{2}|V(Y)| \le \frac{n}{8k}+\frac{1}{2}$. Thus, $|P_i|\ge c$. By Hall's Marriage Theorem, we can find distinct edges $e_1, \dots, e_c$ in $H$ such that $e_i \in P_i$. Thus, we can construct a Hamiltonian cycle of $M$ from $H$ by connecting one endpoint of $e_i$ to $a_i$ and the other to $b_i$, adding the path from $a_i$ to $b_i$ in $Y$, and removing the edge $e_i$ for each $i$.
\end{proof}

\begin{claim}\label{Twocycle}
At any time in Phase 1, $M$ contains $k$ disjoint cycles of length $\frac{n}{k}$ whose union contains $Y$.
\end{claim}

\begin{proof}
Partition the vertices in $M'$ into $k$ subsets $V_1, \dots, V_k$ such that $|V_j| = \frac{n}{k}$ for each $j = 1, \dots, k-1$ and $|V_k| = |V(M')| - \frac{(k-1)n}{k} = \frac{n}{k} - |V(Y)|$. Let $M_1, \dots, M_k$ be the subgraphs of $M'$ induced by $V_1, \dots, V_k$ respectively. Now let $j = 1, \dots, k-1$. By Claim \ref{Mdegree}, in $M_j$, the degree of each vertex is at least
\[  \frac{(4k-1)n}{4k} - |V(Y)| -1 - \bigg(\frac{(k-1)n}{k} - |V(Y)|\bigg) = \frac{3n}{4k} -1 \ge \frac{1}{2}|V_j|.  \]
By Dirac's Theorem, we can find a Hamiltonian cycle $H_j$ of $M_j$. Note that the cycles $H_1, \dots, H_{k-1}$ are disjoint cycles in $M$ with length $\frac{n}{k}$.

In $M_k$, the degree of each vertex is at least
\[  \frac{(4k-1)n}{4k} - |V(Y)| -1 - \frac{(k-1)n}{k} = \frac{3n}{4k} - |V(Y)| -1 \ge \frac{1}{2}\bigg(\frac{n}{k} - |V(Y)|\bigg) = \frac{1}{2}|V_k|,   \]
where for the inequality we again used that $|V(Y)| \le \frac{n}{4k}+1$ in Phase 1. By Dirac's Theorem, we can find a Hamiltonian cycle $H_k$ of $M_k$. We now modify $H_k$ to construct a cycle on the vertices $V(Y) \cup V_k$ using a similar construction as in the proof of Claim \ref{Hcycle}. For each $i = 1, \dots, c$, by Claim \ref{Mdegree}, each of $a_i, b_i$ is adjacent in $M$ to at least
\[  \frac{(4k-1)n}{4k}-|V(Y)| -1 - \frac{(k-1)n}{k} = \frac{3n}{4k} - |V(Y)|-1     \] vertices in $M_k$. Let $P'_i$ denote the set of edges in $H_k$ that has one endpoint adjacent to $a_i$ and the other endpoint adjacent to $b_i$. By Claim \ref{edgecount} applied to the cycle $H_k$ and the two subsets of vertices that are adjacent to $a_i$, $b_i$ respectively, we see that
\[  |P'_i| \ge 2 \bigg(\frac{3n}{4k}-|V(Y)|-1 \bigg) - \bigg(\frac{n}{k}- |V(Y)|\bigg) = \frac{n}{2k}- |V(Y)|-2 \ge \frac{n}{4k} -3.\]
As noted before, we have $c \le \frac{n}{8k}+\frac{1}{2}$ in Phase 1. Thus $|P'_i| \ge c$. By Hall's Marriage Theorem, we can find distinct edges $e'_1, \dots, e'_c$ in $H_k$ such that $e'_i \in P'_i$. Thus, we can construct a cycle on $V(Y) \cup V_k$ from $H_k$ by connecting one endpoint of $e'_i$ to $a_i$ and the other to $b_i$, adding the path from $a_i$ to $b_i$ in $Y$, and removing the edge $e'_i$ for each $i$. This cycle is disjoint from $H_1, \dots, H_{k-1}$ in $M$ and has length $\frac{n}{k}$.
\end{proof}

\section{Randomized Round Lower Bound for $\OCTC$ via Approximate Certificate Complexity}\label{RandomBounds}
In this section, we focus on \emph{randomized} round complexities of partial Boolean functions in the $\ampc$ model. By combining results from the previous sections with Yao's Lemma \cite{Yao}, we relate the number of rounds required by a randomized $\ampc$ algorithm to compute a partial Boolean function to its \emph{approximate certificate complexity}. In addition, we develop new machinery for proving lower bounds on the approximate certificate complexity for a general class of partial functions. This machinery in particular leads to an asymptotically optimal lower bound on the approximate certificate complexity of $\OCTC$, which we use to prove an $\Omega(\log_S n)$ randomized round lower bound for computing $\OCTC$ in $\ampc$. If $S = n^\epsilon$ for $\epsilon \in (0,1)$, our round lower bound matches the upper bound of Behnezhad et al. \cite{Behnezhad}.

\subsection{Efficient Randomized $\ampc$ Algorithms Imply Small Approximate Certificate Complexity}
To formally define the approximate certificate complexity of a partial Boolean function, we first associate to it the following natural distribution on its domain:
\begin{definition}[Canonical distribution over domain of partial function]\label{CanonicalDist}
Let $g : \Delta \to \zo$ be a partial Boolean function with domain $\Delta \subseteq \zo^N$.  We associate with $g$ the distribution $\mathcal{D}_g$ over $\Delta$ defined as follows: with probability $\frac{1}{2}$ output a uniform random assignment from $g^{-1}(1)$, and with probability $\frac{1}{2}$ output a uniform random assignment from $g^{-1}(0)$. 
\end{definition}

\begin{example}[Canonical distribution over $\Delta_\OCTC$]\label{CanonicalOCTC}
For the $\OCTC$ problem on $n$ vertices, the total number of 1-cycle and 2-cycle instances are
\[  n_1 = \frac{(n-1)!}{2}, \qquad n_2 = \frac{1}{2}{n \choose n/2}\left(\frac{(n/2-1)!}{2}\right)^2 = \frac{(n-1)!}{2n} \]
respectively. In the canonical distribution $\D_{\OCTC}$, each 1-cycle instance has probability $\frac{1}{2n_1} = \frac{1}{(n-1)!}$ and each 2-cycle instance has probability $\frac{1}{2n_2}= \frac{n}{(n-1)!}$.
\end{example}

\begin{definition}[Approximate certificate complexity]\label{ApproxCertif}
Let $g: \Delta \to \zo$ be a partial Boolean function with domain $\Delta \subseteq \zo^N$. A \emph{certificate} of $g$ on an input $x \in \Delta$ is a set $C \subseteq \{1, \dots, N\}$ such that for any $y \in \Delta$ satisfying $y|_C = x|_C$, we have $g(y) = g(x)$. The \emph{certificate complexity} of $g$ is defined to be
\[  C(g) := \max_{x \in \Delta} \min \{|C| \mid C \textrm{ is a certificate on }x\}.        \]
Let $\delta \in [0,1)$. We define the \emph{$\delta$-approximate certificate complexity} of $g$ to be 
\[ C_\delta(g) \coloneqq \min\{ C(f)\, \mid\, f : \Delta \to \zo \text{ where $\Prx_{x \sim \mathcal{D}_g}[f(x)\ne g(x)]\le \delta$}. \} \] 
\end{definition}
In particular, we have $C_0(g)=C(g)$.

\begin{example}[Certificate complexity of $\OCKC$]\label{Certif1CkC}
It is not hard to see that $C(\OCKC) = \frac{n}{k}$. A minimum certificate for any 1-cycle instance is a path of length $\frac{n}{k}$ on the cycle, and a minimum certificate for any $k$-cycle instance is one of the $k$ cycles, which involves $\frac{n}{k}$ edges.
\end{example}

Note that for any partial Boolean function $g$, we have $C(g) \le D(g)$. Then Theorem \ref{AMPCDT} immediately implies the following relation between the deterministic round complexity of $g$ in $\ampc$ and $C(g)$:
\begin{corollary}\label{AMPCcertif}
For any $R$-round deterministic $\ampc$ algorithm that computes a partial Boolean function $g: \Delta \to \{0,1\}$, we have
\[  C(g) \le 2S^{6R}. \]
In particular, $R \ge \frac{1}{6} \log_S \frac{C(g)}{2}$.
\end{corollary}

\begin{remark}\rm{
This result combined with Example \ref{Certif1CkC} yields a deterministic round lower bound for computing $\OCKC$ in $\ampc$ similar to that in Theorem \ref{1C2Cdeterministic}. For $\OCTC$, the bound we obtain here has a slightly worse constant.
}
\end{remark}

We now use Corollary \ref{AMPCcertif} to prove the following generalization in the randomized setting, which relates the randomized round complexity of computing a partial Boolean function in $\ampc$ to its approximate certificate complexity. Randomized $\ampc$ algorithms are defined in Definition \ref{RandomAMPC}.

\begin{theorem} \label{AMPCappoxcertif}
Let $\delta \in [0,1)$. For any $R$-round randomized $\ampc$ algorithm that computes a partial Boolean function $g: \Delta \to \{0,1\}$ with error at most $\delta$, we have
\[  C_\delta(g) \le 2S^{6R}.    \]
In particular, $R \ge \frac{1}{6} \log_S \frac{C_{\delta}(g)}{2}$.
\end{theorem}

Note that Corollary \ref{AMPCcertif} can be viewed as a special case of Theorem \ref{AMPCappoxcertif}, since setting $\delta = 0$ in the later recovers the former. Our main tool for proving Theorem \ref{AMPCappoxcertif} is Yao's Lemma specialized to the $\ampc$ model:

\begin{lemma}[Yao \cite{Yao}] \label{YaoLemma}
Let $g: \Delta \to \zo$ be a partial Boolean function. Suppose there exists a distribution $\D$ on $\Delta$ such that any deterministic $\ampc$ algorithm $\A$ computing a partial Boolean function on $\Delta$ with
\[  \Prx_{x \sim \D}[\mbox{$\A$ does not return $g(x)$ on $x$}] \le \delta  \]
requires at least $K$ rounds. Then any randomized $\ampc$ algorithm for $g$ with error at most $\delta$ also requires at least $K$ rounds.
\end{lemma}

\begin{proof}[Proof of Theorem \ref{AMPCappoxcertif}]
We will use Yao's Lemma (\ref{YaoLemma}) together with the canonical distribution $\D_g$ of $g$. Let $\A$ be an $R'$-round deterministic $\ampc$ algorithm that computes a partial Boolean function $f:\Delta \to \zo$ with the property that
\[  \Pr_{x \sim \D_g}[f(x) \neq g(x)] \le \delta.  \]
By Corollary \ref{AMPCcertif} applied to $f$ and Definition \ref{ApproxCertif}, we have
\[  2S^{6R'} \ge C(f) \ge C_\delta(g).    \]
That is, $R' \ge \frac{1}{6}\log_S \frac{C_\delta(g)}{2}$. We can then conclude by Yao's Lemma (\ref{YaoLemma}).
\end{proof}

\subsection{A Method for Lower-bounding Approximate Certificate Complexity}
Theorem \ref{AMPCappoxcertif} reduces randomized round lower bounds for computing a partial Boolean function in $\ampc$ to lower bounds on its approximate certificate complexity. However, even for the $\OCTC$ problem, it is still challenging to lower bound its $\delta$-approximate certificate complexity for a positive $\delta$. We tackle this challenge by developing a new method for proving lower bounds on the approximate certificate complexities of certain partial Boolean functions, which is based on sensitive blocks. In particular, our method applies to $\OCTC$.

\begin{definition}[Sensitive block]
Let $g: \Delta \to \zo$ be a partial Boolean function with domain $\Delta \subseteq \zo^N$, and let $x \in \Delta$. A subset $B \subseteq [1, \dots, N]$ is a \emph{sensitive block} of $g$ on $x$ if the input $x^B$ obtained from $x$ by flipping every bit in $B$ is also contained in $\Delta$ and satisfies $g(x^B) \neq g(x)$.
\end{definition}

\begin{theorem}[Framework for proving lower bounds on approximate certificate complexity] \label{ApproxCertifLB}
Let $g: \Delta \to \zo$ be a partial Boolean function and $\delta \in (0,1/2)$. Denote $V_1 = g^{-1}(1)$ and $V_0 = g^{-1}(0)$ as the set of 1- and 0-instances of $g$ respectively. Suppose that for some number $K$, there is a bipartite graph $H= (V_1 \cup V_0, E)$ that satisfies:
\begin{itemize}
    \item[$\circ$] For any edge $(x,y)$ between $x \in V_1$ and $y \in V_0$ in $E$, there is a sensitive block $B$ of $g$ on $x$ such that $y = x^{B}$.
    \item[$\circ$] The degree of any $x \in V_1$, denoted by $\deg(x)$, is at least $2K$, and the sensitive blocks $B_{x,1}, \dots,$ $B_{x,\deg(x)}$ on $x$ corresponding to its incident edges are all disjoint.
    \item[$\circ$] The degree of any $y \in V_0$ is at most $d = (\frac{1}{2\delta} -1)K \frac{|V_1|}{|V_0|}$.
\end{itemize}
Then
\[  C_\delta(g) \ge K.    \]
\end{theorem}
The bipartite graph $H$ in Theorem \ref{ApproxCertifLB} is illustrated in~\Cref{fig:ApproxCertifConstruction}.

\begin{figure}[htb]
\begin{center}
\includegraphics[scale=.4]{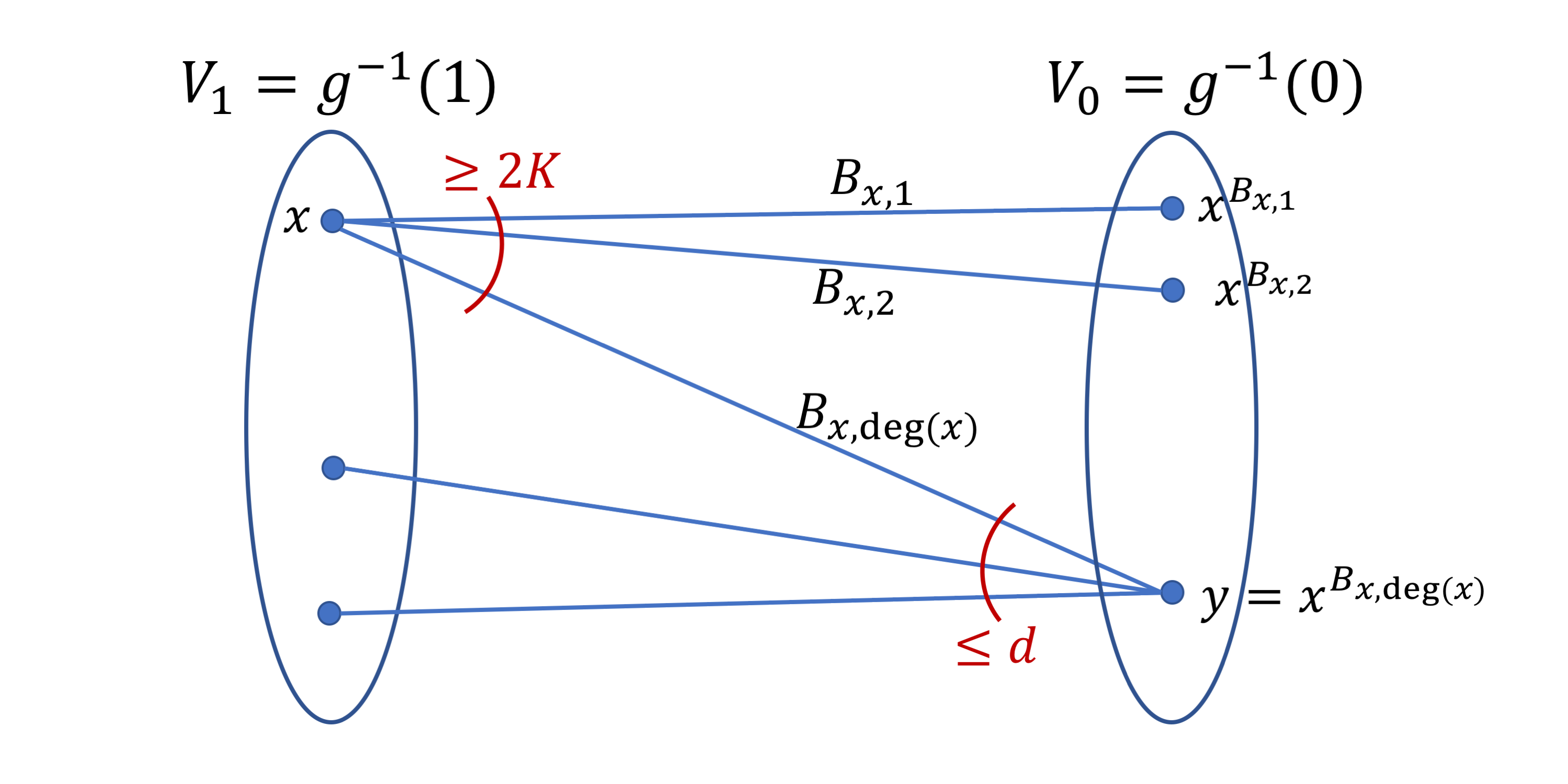}
\end{center}
\vspace{-20pt}
\caption{The bipartite graph $H$ in Theorem \ref{ApproxCertifLB}. The neighbors of each 1-instance $x \in V_1$ correspond to the 0-instances obtained by flipping the chosen sensitive blocks $B_{x_1}, \dots, B_{x, \deg(x)}$ on $x$, and there are at least $2K$ of them. Each 0-instance $y \in V_0$ has at most $d$ neighbors.}
\label{fig:ApproxCertifConstruction}
\end{figure}

\begin{proof}
We need to prove that for any partial Boolean function $f:\Delta \to \zo$ with $C(f) < K$, we have
\[  \Prx_{x \sim \D_g}[f(x) \neq g(x)] > \delta,    \]
where $\D_g$ is the canonical distribution over the domain $\Delta$ of $g$ defined in Definition \ref{CanonicalDist}.

Take a partial Boolean function $f$ as above, and let $U \subseteq V_1$ be the set of 1-instances of $g$ on which $f$ agree with $g$. We consider two cases. If $|U| < (1-2\delta)|V_1|$, then $f$ disagrees with $g$ on more than $2\delta|V_1|$ many 1-instances. Since each 1-instance has probability $\frac{1}{2|V_1|}$ in $\D_g$, we have
\[  \Pr_{x \sim \D_g}[f(x) \neq g(x)] > 2\delta|V_1| \cdot \frac{1}{2|V_1|} = \delta.  \]

Suppose otherwise that $|U| \ge (1-2\delta)|V_1|$. We show in this case that there are a significant number of 0-instances on which $f$ and $g$ disagree. Fix $x \in U$. Since $C(f) < K$, we can find a certificate $C_x$ of $f$ on $x$ that involves fewer than $K$ input bits. Note that if a sensitive block $B$ of $g$ on $x$ is disjoint from $C_x$, then
\[  f(x^B) = f(x) = g(x) \neq g(x^B). \]
In other words, $f$ disagrees with $g$ on $x^B$.

Since the sensitive blocks $B_{x,1}, \dots, B_{x,\deg(x)}$ of $g$ on $x$ given by the theorem are all disjoint, $C_x$ intersects fewer than $K$ of them. Moreover, since $\deg(x) \ge 2K$, $C_x$ is disjoint from more than $2K-K = K$ of these sensitive blocks. This implies that $x$ is adjacent in $H$ to more than $K$ 0-instances on which $f$ disagrees with $g$.

Now since each 0-instance is adjacent to at most $d$ 1-instances in $H$, the number of 0-instances on which $f$ and $g$ disagree is greater than
\[  \frac{|U| \cdot K}{d} \ge \frac{(1-2\delta)|V_1| \cdot K}{(\frac{1}{2\delta}-1)K\frac{|V_1|}{|V_0|}} = 2\delta|V_0|.  \]
Finally, since each 0-instance has probability $\frac{1}{2|V_0|}$ in $\D_g$, we have
\[  \Pr_{x \sim \D_g}[f(x) \neq g(x)] > 2\delta|V_0| \cdot \frac{1}{2|V_0|} = \delta, \]
and the proof is complete.
\end{proof}

\begin{example}[Approximate certificate complexity of a promised \textsc{Majority} problem]
We illustrate the conditions of Theorem \ref{ApproxCertifLB} by the following simple promise version of the {\sc Majority} problem: Let $N = 2N'+1$ be an odd integer, and $\Delta \subset \zo^N$ be the subset of $N$-bit strings whose Hamming weight is $N'$ or $N'+1$. Let $f: \Delta \to \zo$ be the restriction of the {\sc Majority} function to $\Delta$:
\[  f(x) = \begin{cases}
        1 & \mbox{if the Hamming weight of $x$ is $N'+1$}\\
        0 & \mbox{if the Hamming weight of $x$ is $N'$}.
    \end{cases} \]
Note that $V_1 = f^{-1}(1)$ and $V_0 = f^{-1}(0)$ have the same size. We take $\delta = \frac{1}{6}$ and $K = \frac{N'+1}{2}$. Then
\[  d = \Big(\frac{1}{2\delta}-1\Big)K\frac{|V_1|}{|V_0|} = N'+1.   \]
For $x \in V_1$, each 1-bit itself forms a sensitive block of $f$ on $x$, and thus the $N'+1$ many 1-bits give us $N'+1 = 2K$ disjoint sensitive blocks. Moreover, each input in $V_0$ differs by a bit from exactly $N'+1 = d$ inputs in $V_1$. Thus Theorem \ref{ApproxCertifLB} applies and shows that
\[  C_{\frac{1}{6}}(f) \ge \frac{N'+1}{2} = \frac{1}{2}C(f).    \]
\end{example}

\subsection{Randomized Round Complexity of $\OCTC$}
As a more involved example, we now apply Theorem \ref{ApproxCertifLB} to prove the following lower bound on the approximate certificate complexity of $\OCTC$:
\begin{corollary}[Approximate certificate complexity of $\OCTC$]\label{1C2CApproxCertif}
\[  C_{\frac{1}{6}}(\OCTC) \ge \frac{n}{4}. \]
\end{corollary}
Note that this is asymptotically optimal, since
\[  C_{\frac{1}{6}}(\OCTC) \le C(\OCTC) = \frac{n}{2}.     \]

\begin{proof}
For $\delta = \frac{1}{6}$ and $K = \frac{n}{4}$, we will find sensitive blocks of $\OCTC$ on the 1-cycle instances that induce a bipartite graph satisfying the conditions in Theorem \ref{ApproxCertifLB}. Recall from Example \ref{CanonicalOCTC} that the number of 1-cycle and 2-cycle instances are $n_1 = \frac{(n-1)!}{2}$ and $n_2=\frac{(n-1)!}{2n}$ respectively. Then
\[  d = \Big(\frac{1}{2\delta}-1\Big)K\frac{n_1}{n_2} = \frac{n^2}{2}.   \]

Let $x \in V_1$ be a 1-cycle instance. If we delete two opposite edges on the cycle $x$ and add the two edges that complete the two remaining paths into cycles, we obtain a 2-cycle instance. (See~\Cref{fig:OCTCApproxCertifConstruction}.) Thus the four edges above together form a sensitive block of $\OCTC$ on $x$. Note that there are $\frac{n}{2} = 2K$ pairs of opposite edges on $x$, and the induced sensitive blocks are all disjoint.

\begin{figure}[htb]
\begin{center}
\includegraphics[scale=.42]{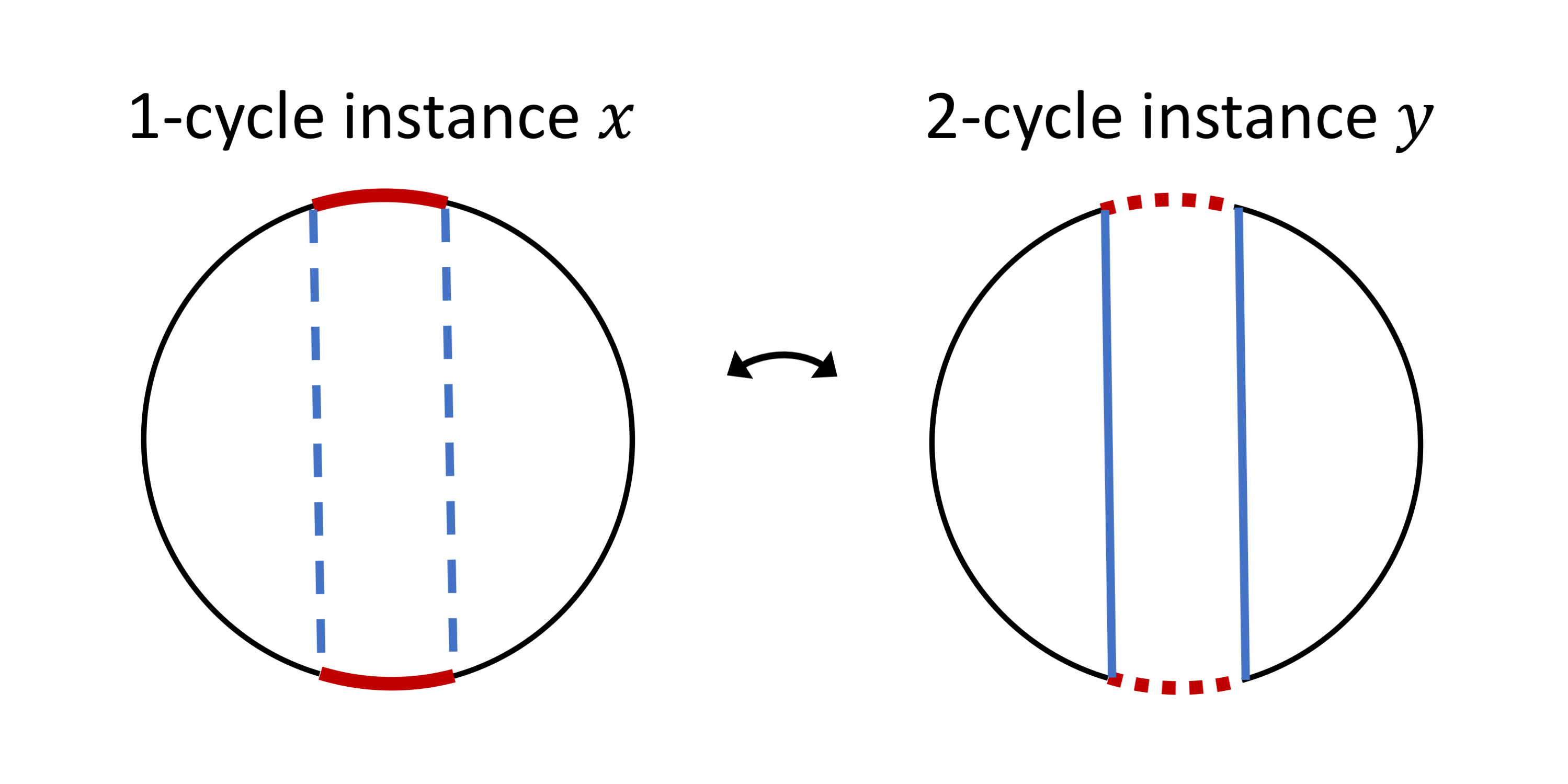}
\end{center}
\vspace{-20pt}
\caption{Converting a 1-cycle instance $x$ into a 2-cycle instance $y$, and vice versa.}
\label{fig:OCTCApproxCertifConstruction}
\end{figure}
\medskip

Now we verify that for any 2-cycle instance $y \in V_0$, $y$ is obtained from at most $d$ 1-cycle instances via the process above. To reverse this process, we need to first remove one edge from each of the two cycles in $y$ to obtain two disjoint paths $P_1$ and $P_2$, and then connect one endpoint of $P_1$ to one endpoint of $P_2$, and the other endpoint of $P_1$ to the other endpoint of $P_2$. We have a choice of $n/2$ edges on each cycle of $y$ to delete, and the there are two ways of connecting the two remaining paths. Thus the number of 1-cycle instances that can be modified into $y$ by flipping a sensitive block as above is exactly
\[   2 \cdot \frac{n}{2} \cdot \frac{n}{2} = \frac{n^2}{2} = d.  \]
Thus we can apply Theorem \ref{ApproxCertifLB} to obtain our desired lower bound.
\end{proof}

With Corollary \ref{1C2CApproxCertif} and Theorem \ref{AMPCappoxcertif}, we are now ready to prove the desired lower bound on the randomized round complexity of $\OCTC$ in the $\ampc$ model:

\begin{theorem}\label{1C2Crandom}
Any randomized $\ampc$ algorithm that computes $\OCTC$ with error at most $1/6$ requires at least $\frac{1}{6}\log_S n - \frac{1}{2} \log_S 2 = \Omega(\log_S n)$ rounds. In particular, if $S = n^\epsilon$ for $\epsilon \in (0,1)$, then any such randomized $\ampc$ algorithm requires $\Omega(1/\epsilon)$ rounds.
\end{theorem}

\begin{proof}
Given an $R$-round randomized $\ampc$ algorithm that computes $\OCTC$ with error at most $\frac{1}{6}$, Theorem \ref{AMPCappoxcertif} and Corollary \ref{1C2CApproxCertif} together imply that
\[ 2S^{6R} \ge C_{\frac{1}{6}}(\OCTC) \ge \frac{n}{4}.  \]
This implies that $R \ge\frac{1}{6}\log_S n - \frac{1}{2} \log_S 2,$
and the proof is complete.
\end{proof}

\section*{Acknowledgments}
We thank Jakub \L\k{a}cki and the co-authors of \cite{Behnezhad} for clarifications on the $\ampc$ model and valuable feedback on our specifications of the model.  

Moses Charikar was supported by a Simons Investigator Award, a Google Faculty Research Award and an Amazon Research Award.  Weiyun Ma was supported by a Stanford Graduate Fellowship.  Li-Yang Tan was supported by NSF grant CCF-1921795.

\bibliographystyle{alpha}
\bibliography{ref}

\end{document}